\documentclass[11pt]{llncs}
\usepackage{amsmath}
\usepackage{amssymb}
\usepackage{amsfonts}
\usepackage{epic,gastex}
\usepackage{eucal}
\usepackage{array}
\usepackage{url}

\newcommand{\dt}{.}

\DeclareMathOperator{\excl}{\mathrm{excl}}
\DeclareMathOperator{\dupl}{\mathrm{dupl}}

\newcommand{\cra}{completely reachable automata}
\newcommand{\cran}{completely reachable automaton}

\newcommand{\scc}{SCC}
\newcommand{\scn}{strongly connected}

\DeclareSymbolFont{rsfscript}{OMS}{rsfs}{m}{n}
\DeclareSymbolFontAlphabet{\mathrsfs}{rsfscript}

\begin{document}
\title{A Characterization\\ of Completely Reachable Automata\thanks{Supported by the Russian Foundation for Basic Research, grant no.\ 16-01-00795, the Russian Ministry of Education and Science, project no.\ 1.3253.2017, and the Competitiveness  Enhancement Program of Ural Federal University.}}

\titlerunning{Completely Reachable Automata}

\author{E. A. Bondar \and M. V. Volkov}

\authorrunning{E. A. Bondar, M. V. Volkov}

\tocauthor{E. A. Bondar, M. V. Volkov (Ekaterinburg, Russia)}

\institute{Institute of Natural Sciences and Mathematics\\ Ural Federal University, Lenina 51, 620000 Ekaterinburg, Russia\\
\email{bondareug@gmail.com, mikhail.volkov@usu.ru}}

\maketitle

\begin{abstract}
A complete deterministic finite automaton in which every non-empty subset of the state set occurs as the image of the whole state set under the action of a suitable input word is called completely reachable. We characterize completely reachable automata in terms of certain directed graphs.

\end{abstract}

\section{Overview}
A \emph{complete deterministic finite automaton} (DFA) is a pair $\langle Q,\Sigma\rangle$, where $Q$ and $\Sigma$ are finite sets called the \emph{state set} and the \emph{input alphabet} respectively, together with a totally defined map $Q\times\Sigma\to Q$, $(q,a)\mapsto q\dt a$, called the \emph{transition function}. Let $\Sigma^*$ stand for the collection of all finite words over the alphabet $\Sigma$, including the empty word. The transition function extends to a function $Q\times\Sigma^*\to Q$ in the following natural way: for every $q\in Q$ and $w\in\Sigma^*$, we set $q\dt w:=q$ if $w$ is empty and $q\dt w:=(q\dt v)\dt a$ if $w=va$ for some $v\in\Sigma^*$ and some $a\in\Sigma$. Thus, every word $w\in\Sigma^*$ induces a transformation of the set $Q$. If $P$ is a non-empty subset of $Q$, we write $P\dt w$ for $\{q\dt w\mid q\in P\}$.

Given a DFA $\mathrsfs{A}=\langle Q,\Sigma\rangle$, we say that a non-empty subset $P\subseteq Q$ is \emph{reachable} in $\mathrsfs{A}$ if $P=Q\dt w$ for some word $w\in\Sigma^*$. A DFA is called \emph{completely reachable} if every non-empty subset of its state set is reachable.

Our paper \cite{cra1} lists several motivations for considering completely reachable automata. Here we only mention that they have appeared in the study of descriptional complexity of formal languages~\cite{Maslennikova:2012,cra1} and in relation to the \v{C}ern\'{y} conjecture~\cite{Don:2016,GGGJV:2017}. In the present note we characterize completely reachable automata in terms of certain directed graphs.

The paper is organized as follows. In Section~2 we recall a sufficient condition for complete reachability that has been established in~\cite{cra1}. The graph that appears in this condition serves as a departure point for an iterative construction that we describe in Section~3. Our main result is stated and proved in Section~4, followed by a final discussion in Section~5.

\section{The graph $\Gamma_1(\mathrsfs{A})$}

Since we consider only directed graphs in this paper, we call them just graphs in the sequel. Given a graph $\Gamma$, a vertex $p$ is said to be \emph{reachable} from a vertex $q$ if there exists a directed path starting at $q$ and terminating at $p$. The \emph{reachability} relation on $\Gamma$ consists of all pairs $(p,q)$ of vertices such that either $p=q$ or $p$ is reachable from $q$. Clearly, the reachability relation is a pre-order on the vertex set of $\Gamma$, and the equivalence corresponding to the pre-order partitions this set into classes of mutually reachable vertices. The subgraphs of $\Gamma$ induced on these classes are called the \emph{strongly connected components} (SCCs, for short) of $\Gamma$. A graph with a unique \scc\ is said to be \emph{strongly connected}.

If $\mathrsfs{A}=\langle Q,\Sigma\rangle$ is a DFA, the \emph{defect} of a word $w\in\Sigma^*$ with respect to $\mathrsfs{A}$ is defined as the size of the set $Q{\setminus}Q\dt w$. Consider a word $w$ of defect~1. For such a word, the set $Q{\setminus}Q\dt w$ consists of a unique state, which is called the \emph{excluded state} for $w$ and is denoted by $\excl(w)$. Further, the set $Q\dt w$ contains a unique state $p$ such that $p=q_1\dt w=q_2\dt w$ for some $q_1\ne q_2$; this state $p$ is called the \emph{duplicate state} for $w$ and is denoted by $\dupl(w)$. Let $W_1(\mathrsfs{A})$ stand for the set of all words of defect~1 with respect to $\mathrsfs{A}$, and let $\Gamma_1(\mathrsfs{A})$ denote the graph with the vertex set $Q$ and the edge set
\begin{equation}
\label{eq:e1}
E_1:=\{(\excl(w),\dupl(w))\mid w\in W_1(\mathrsfs{A})\}.
\end{equation}
We say that the edge $(\excl(w),\dupl(w))\in E_1$ is \emph{enforced} by the word $w$.

\begin{theorem}[\mdseries\cite{cra1}]
\label{thm:sufficient}
If a DFA $\mathrsfs{A}=\langle Q,\Sigma\rangle$ is such that the graph $\Gamma_1(\mathrsfs{A})$ is \scn, then $\mathrsfs{A}$ is completely reachable; more precisely, for every non-empty subset $P\subseteq Q$, there is a product $w$ of words of defect~$1$ such that $P=Q\dt w$.
\end{theorem}

In Theorem~\ref{thm:sufficient} and in similar statements below we do not exclude the case when $P=Q$ since we may consider the empty word as the product of the empty set of factors with any prescribed property.

The following example, also taken from~\cite{cra1}, demonstrates that the condition of Theorem~\ref{thm:sufficient} is not necessary.
\begin{example}
\label{examp:not necessary}
Consider the DFA $\mathrsfs{E}_3$ with the state set $\{1,2,3\}$ and the input letters $a_{[1]},a_{[2]},a_{[3]},a_{[1,2]}$ that act as follows:
\begin{eqnarray*}
i\dt a_{[1]}:=\begin{cases}
2 &\text{if } i=1,2,\\
3 &\text{if } i=3;
\end{cases}
&\qquad&
i\dt a_{[2]}:=\begin{cases}
1 &\text{if } i=1,2,\\
3 &\text{if } i=3;
\end{cases}
\\
i\dt a_{[3]}:=\begin{cases}
1 &\text{if } i=1,2,\\
2 &\text{if } i=3;
\end{cases}
&\qquad&
i\dt a_{[1,2]}:=3\ \text{ for all }\ i=1,2,3.
\end{eqnarray*}

\begin{figure}[bt]
\begin{center}
\unitlength=1.05mm
\begin{picture}(80,32)(-7.5,-5)
\node(A1)(0,0){1} \node(B1)(24,0){2} \node(C1)(12,17){3}
\drawedge[ELside=r,curvedepth=-3](A1,B1){$a_{[1]}$}
\drawedge[ELside=r,ELpos=55](B1,A1){$a_{[2]},a_{[3]}$}
\drawedge[ELside=r,ELpos=45](C1,B1){$a_{[3]}$}
\drawedge[ELside=r,curvedepth=-3](B1,C1){$a_{[1,2]}$}
\drawedge[curvedepth=3](A1,C1){$a_{[1,2]}$}
\drawloop[ELdist=.8,loopangle=0](B1){$a_{[1]}$}
\drawloop[ELdist=.8,loopangle=180](A1){$a_{[2]},a_{[3]}$}
\drawloop[ELdist=.8](C1){$a_{[2]},a_{[1]},a_{[1,2]}$}
\node(A2)(55,0){1} \node(B2)(79,0){2} \node(C2)(67,17){3}
\drawedge[ELside=r,curvedepth=-3](A2,B2){$a_{[1]}$}
\drawedge[ELside=r](B2,A2){$a_{[2]}$}
\drawedge[ELside=r](C2,A2){$a_{[3]}$}
\end{picture}
\caption{The automaton $\mathrsfs{E}_3$ and the graph $\Gamma_1(\mathrsfs{E}_3)$. Each edge of $\Gamma_1(\mathrsfs{E}_3)$ is labeled by a word of defect~1 that enforces this edge.}\label{fig:e3}
\end{center}
\end{figure}
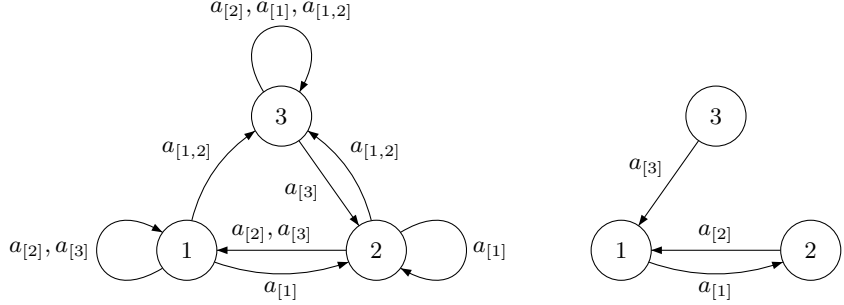
The automaton $\mathrsfs{E}_3$ is shown in Fig.\,\ref{fig:e3} on the left. The graph $\Gamma_1(\mathrsfs{E}_3)$ is shown in Fig.\,\ref{fig:e3} on the right; it is not \scn. However, it can be checked by a straightforward computation that the automaton $\mathrsfs{E}_3$ is completely reachable. (This will also follow from Theorem~\ref{thm:sufficient2} below.)
\end{example}

\section{The graph $\Gamma_k(\mathrsfs{A})$ for $k>1$}

In order to generalize Theorem~\ref{thm:sufficient}, we first extend the operators $\excl(\_)$ and $\dupl(\_)$ to words with defect $>1$. Namely, if $\mathrsfs{A}=\langle Q,\Sigma\rangle$ is a DFA and $w\in\Sigma^*$, we define $\excl(w)$ as the set $Q{\setminus}Q\dt w$ and $\dupl(w)$ as the set $\{p\in Q\mid p=q_1\dt w=q_2\dt w \ \text{ for some }\ q_1\ne q_2\}$. If we take the usual liberty of ignoring the distinction between singleton sets and their elements, then for words of defect~1, the new meanings of $\excl(w)$ and $\dupl(w)$ agree with the definition from~\cite{cra1}.

Now we describe an iterative process that assigns to each given DFA $\mathrsfs{A}=\langle Q,\Sigma\rangle$ a certain ``layered'' graph $\Gamma(\mathrsfs{A})$. The process starts with the graph $\Gamma_1(\mathrsfs{A})$ defined above. If the graph $\Gamma_1(\mathrsfs{A})$ is \scn, then $\Gamma(\mathrsfs{A}):=\Gamma_1(\mathrsfs{A})$ and the process stops with SUCCESS. If all \scc{}s of $\Gamma_1(\mathrsfs{A})$ are singletons, we also set $\Gamma(\mathrsfs{A}):=\Gamma_1(\mathrsfs{A})$ and the process stops with FAILURE. Except for these two extreme cases, we extend the graph $\Gamma_1(\mathrsfs{A})$ as follows. Let $Q_2$ be the collection of the vertex sets of all at least 2-element \scc{}s of the graph $\Gamma_1(\mathrsfs{A})$ and let $W_2(\mathrsfs{A})$ stand for the set of all words of defect~2 with respect to $\mathrsfs{A}$. We define $\Gamma_2(\mathrsfs{A})$ as the graph whose vertex set is $Q\cup Q_2$ and whose edge set is the union of $E_1$ with the set $I_2:=\{(q,C)\in Q\times Q_2\mid q\in C\}$ of \emph{inclusion} edges representing the containments between the elements of $Q$ and the \scc{}s in $Q_2$ and the set
\begin{multline}
\label{eq:e2}
E_2:=\{(C,p)\in Q_2\times Q\mid C\supseteq\excl(w),\, p\in\dupl(w)\\ \text{for some }\ w\in W_2(\mathrsfs{A})\}.
\end{multline}
Extending the terminology used for edges from $E_1$, we say that the edge $(C,p)\in E_2$ with $C\supseteq\excl(w)$, $p\in\dupl(w)$ is \emph{enforced} by the word $w$.

Observe that the definition of $E_1$ in \eqref{eq:e1} can be easily restated in the form similar to the one of the definition of $E_2$ in \eqref{eq:e2}:
\begin{multline*}
E_1:=\{(q,p)\in Q\times Q\mid \{q\}\supseteq\excl(w),\, p\in\dupl(w)\\ \text{for some }\ w\in W_1(\mathrsfs{A})\}.
\end{multline*}

For an illustration, see Fig.\,\ref{fig:g2e3} that displays the graph $\Gamma_2(\mathrsfs{E}_3)$,
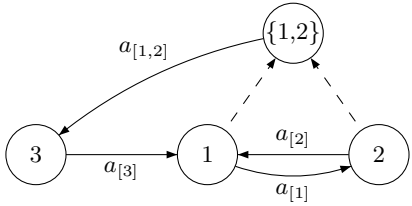
\begin{figure}[t]
\begin{center}
\unitlength=.95mm
\begin{picture}(50,22)(0,-5)
\node(A2)(25,0){1}
\node(B2)(49,0){2}
\node(C2)(1,0){3}
\drawedge[ELside=r,curvedepth=-3](A2,B2){$a_{[1]}$}
\drawedge[ELside=r](B2,A2){$a_{[2]}$}
\drawedge[ELside=r](C2,A2){$a_{[3]}$}
\node(C12)(37,17){\{1,2\}}
\drawedge[dash={1.5}{1.5}](A2,C12){}
\drawedge[dash={1.5}{1.5}](B2,C12){}
\drawedge[ELside=r,curvedepth=-3](C12,C2){$a_{[1,2]}$}
\end{picture}
\caption{The graph $\Gamma_2(\mathrsfs{E}_3)$, in which each non-inclusion edge is labeled by a word of defect~1 or 2 that enforces this edge.}\label{fig:g2e3}
\end{center}
\end{figure}
the inclusion edges being shown with dashed arrows. Observe that the graph is \scn.

Before we proceed with the description of the generic step of our process, we present an intermediate result. Even though it will be superseded by Theorem~\ref{thm:sufficient3}, we believe that its proof may help the reader to better understand the intuition behind our construction and to easier digest more technical arguments that occur later.

\begin{theorem}
\label{thm:sufficient2}
If a DFA $\mathrsfs{A}=\langle Q,\Sigma\rangle$ is such that the graph $\Gamma_2(\mathrsfs{A})$ is \scn, then $\mathrsfs{A}$ is completely reachable; more precisely, for every non-empty subset $P\subseteq Q$, there is a product $w$ of words of defect at most\/ $2$ such that $P=Q\dt w$.
\end{theorem}

\begin{proof}
Take any non-empty subset $P\subseteq Q$. We prove that $P$ is reachable in $\mathrsfs{A}$ via a product of words of defect $\le2$ by induction on $m:=|Q{\setminus} P|$. If $m=0$, then $P=Q$ and nothing is to prove as $Q$ is reachable via the empty word. Now let $m>0$ so that $P$ is a proper subset of $Q$. We aim to find a subset $R\subseteq Q$ such that $P=R\dt w$ for some word $w$ of defect at most~2 and $|R|>|P|$. Then $|Q{\setminus} R|<m$, and the induction assumption applies to the subset $R$ whence $R=Q\dt v$ for some product $v$ of words of defect at most 2. Then $P=Q\dt vw$ so that $P$ is reachable as required.

Consider two cases.

\smallskip

\noindent\textit{\textbf{Case 1:} There exists an edge $(q,p)\in E_1$ such that $q\in Q{\setminus} P$ and $p\in P$}.

In this case, the argument from the proof of Theorem~\ref{thm:sufficient} readily applies; we reproduce it here for the reader's convenience. Since $(q,p)\in E_1$, there is a word $w$ of defect~1 with respect to $\mathrsfs{A}$ for which $q$ is the excluded state and $p$ is the duplicate state. By the definition of the duplicate state, $p=q_1\dt w=q_2\dt w$ for some $q_1\ne q_2$, and since the excluded state $q$ for $w$ does not belong to $P$, for each state $r\in P{\setminus}\{p\}$, there exists a state $r'\in Q$ such that $r'\dt w=r$. Now letting $R:=\{q_1,q_2\}\cup\bigl\{r'\mid r\in P{\setminus}\{p\}\bigr\}$, we conclude that $P=R\dt w$ and $|R|=|P|+1$.

\smallskip

\noindent\textit{\textbf{Case 2:} For every edge $(q,p)\in E_1$, if $q\in Q{\setminus} P$, then also $p\in Q{\setminus} P$}.

Here, it is easy to see that every \scc\ of the graph $\Gamma_1(\mathrsfs{A})$ is either contained in $P$ or disjoint with $P$.
Let $\widehat{P}:=P\cup\{C\in Q_2\mid C\subseteq P\}$. This is a proper subset of $Q\cup Q_2$ as $P$ is a proper subset of $Q$.

Since the graph $\Gamma_2(\mathrsfs{A})$ is \scn, there must exist an edge $e$ that connects $(Q\cup Q_2){\setminus}\widehat{P}$ with $\widehat{P}$ in the sense that the head of the edge $e$ belongs to $(Q\cup Q_2){\setminus}\widehat{P}$ while the tail of $e$ lies in $\widehat{P}$. Under the condition of Case~2, the edge $e$ cannot belong to $E_1$. Furthermore, the definition of $\widehat{P}$ eliminates the possibility for $e$ to be an inclusion edge: if $(q,C)\in I_2$ is such that $C\in\widehat{P}$, then $C\subseteq P$ whence $q\in P$. Thus, we conclude that $e\in E_2$, i.e., $e=(C,p)$ where $C\notin\widehat{P}$ and $p\in P$. By the definition of $E_2$, there exists a word $w$ of defect~2 with respect to $\mathrsfs{A}$ such that $C\supseteq\excl(w)$ and $p\in\dupl(w)$. By the definition of $\dupl(w)$, there exist some $q_1,q_2\in Q$ such that $q_1\ne q_2$ and $q_1\dt w=q_2\dt w=p$. Since $C\notin\widehat{P}$, we have $C\cap P=\varnothing$, whence $\excl(w)\cap P=\varnothing$. Thus, for every state $r\in P{\setminus}\{p\}$, there is a state $r'\in Q$ such that $r'\dt w=r$. Now we can proceed as in Case~1: we set $R:=\{q\mid q\dt w=p\}\cup\bigl\{r'\mid r\in P{\setminus}\{p\}\bigr\}$ and conclude that $P=R\dt w$ and $|R|>|P|$ since $q_1,q_2\in R$.\qed
\end{proof}

Now we return to our iterative definition of $\Gamma(\mathrsfs{A})$. Suppose that $k>2$ and the graph $\Gamma_{k-1}(\mathrsfs{A})$  with the vertex set $Q\cup Q_2\cup\cdots\cup Q_{k-1}$ and the edge set
\begin{equation}
\label{eq:e3}
E_1\cup E_2\cup\cdots\cup E_{k-1}\cup I_2\cup\cdots\cup I_{k-1}
\end{equation}
has already been defined. If the graph $\Gamma_{k-1}(\mathrsfs{A})$ is \scn, then we define $\Gamma(\mathrsfs{A})$ as $\Gamma_{k-1}(\mathrsfs{A})$ and the process stops with SUCCESS. Now suppose that $\Gamma_{k-1}(\mathrsfs{A})$ is not \scn. Given an \scc\ $\Delta$ of $\Gamma_{k-1}(\mathrsfs{A})$, we define its \emph{support} as set of all vertices from $Q$ that belong to $\Delta$ and refer to the cardinality of the support as the \emph{rank} of $\Delta$. If all \scc{}s of $\Gamma_{k-1}(\mathrsfs{A})$ have rank less than $k$, we also set $\Gamma(\mathrsfs{A}):=\Gamma_{k-1}(\mathrsfs{A})$ and the process stops with FAILURE. Otherwise we define the set $Q_k$ as the collection of the supports of all \scc{}s of rank at least $k$ in the graph $\Gamma_{k-1}(\mathrsfs{A})$. Let $W_k(\mathrsfs{A})$ stand for the set of all words of defect~$k$ with respect to $\mathrsfs{A}$. We define $\Gamma_k(\mathrsfs{A})$ as the graph whose vertex set is $Q\cup Q_2\cup\cdots\cup Q_{k-1}\cup Q_k$ and whose edge set is the union of the set \eqref{eq:e3} with the two following sets:
\[
I_k:=\{(q,C)\in Q\times Q_k\mid q\in C\}\cup\bigcup_{i=2}^{k-1}\{(D,C)\in Q_i\times Q_k\mid D\subset C\},
\]
which consists of edges representing the inclusions between the elements of $Q\cup Q_2\cup\cdots\cup Q_{k-1}$ and the \scc{}s in $Q_k$,
\begin{multline}
\label{eq:e4}
E_k:=\{(C,p)\in Q_k\times Q\mid C\supseteq\excl(w),\, p\in\dupl(w)\\ \text{for some }\ w\in W_k(\mathrsfs{A})\},
\end{multline}
which comprises edges enforced by the words of defect $k$ (as in the case $k=2$, we say that the edge $(C,p)\in E_k$ with $C\supseteq\excl(w)$ and $p\in\dupl(w)$ is \emph{enforced} by the word $w$).

The following example illustrates the construction.
\begin{example}
\label{examp:illustration}
Consider the DFA $\mathrsfs{E}_5$ with the state set $\{1,2,3,4,5\}$  and the following transition table:

{\small
\begin{center}
\begin{tabular}{c@{\ }|@{\ }c@{\ }c@{\ }c@{\ }c@{\ }c@{\ }c@{\ }c@{\ }c}
              & $a_{[1]}$ & $a_{[2]}$ & $a_{[3]}$  & $a_{[4]}$ & $a_{[5]}$ & $a_{[1,2]}$ & $a_{[4,5]}$ & $a_{[1,3]}$ \mathstrut\\
\hline
1 & 2 & 1 & 1 & 1 & 1 & 3 & 1 & 4\mathstrut\\
2 & 2 & 1 & 1 & 2 & 2 & 3 & 1 & 4\mathstrut\\
3 & 3 & 3 & 2 & 3 & 3 & 3 & 2 & 4\mathstrut\\
4 & 4 & 4 & 4 & 5 & 4 & 4 & 3 & 5\mathstrut\\
5 & 5 & 4 & 5 & 5 & 4 & 5 & 3 & 5\mathstrut
\end{tabular}\ .
\end{center}

}

\begin{figure}[h!]
\begin{center}
\unitlength=0.95mm
\begin{picture}(105,115)(0,-20)
\put(90,113){$\Gamma_1(\mathrsfs{E}_5)$}
\node(A1)(1,100){2}
\node(B1)(25,100){1}
\node(C1)(49,100){3}
\node(D1)(76,100){4}
\node(E1)(100,100){5}
\drawedge[curvedepth=-3](A1,B1){}
\drawedge(B1,A1){}
\drawedge(C1,B1){}
\drawedge[curvedepth=3](D1,E1){}
\drawedge(E1,D1){}
\put(90,73){$\Gamma_2(\mathrsfs{E}_5)$}
\node(A2)(1,65){2}
\node(B2)(25,65){1}
\node(C2)(49,65){3}
\node(D2)(76,65){4}
\node(E2)(100,65){5}
\drawedge[curvedepth=-3](A2,B2){}
\drawedge(B2,A2){}
\drawedge(C2,B2){}
\drawedge[curvedepth=3](D2,E2){}
\drawedge(E2,D2){}
\node(F2)(13,82){\{1,2\}}
\drawedge[dash={1.5}{1.5}](A2,F2){}
\drawedge[dash={1.5}{1.5}](B2,F2){}
\drawedge[curvedepth=3](F2,C2){}
\node(G2)(88,48){\{4,5\}}
\drawedge[dash={1.5}{1.5}](D2,G2){}
\drawedge[dash={1.5}{1.5}](E2,G2){}
\drawedge[curvedepth=3](G2,C2){}
\drawedge[curvedepth=4](G2,B2){}
\put(90,15){$\Gamma_3(\mathrsfs{E}_5)$}
\node(A3)(1,0){2}
\node(B3)(25,0){1}
\node(C3)(49,0){3}
\node(D3)(76,0){4}
\node(E3)(100,0){5}
\drawedge[curvedepth=-3](A3,B3){}
\drawedge(B3,A3){}
\drawedge(C3,B3){}
\drawedge[curvedepth=3](D3,E3){}
\drawedge(E3,D3){}
\node(F3)(13,17){\{1,2\}}
\drawedge[dash={1.5}{1.5}](A3,F3){}
\drawedge[dash={1.5}{1.5}](B3,F3){}
\drawedge[curvedepth=3](F3,C3){}
\node(G3)(88,-17){\{4,5\}}
\drawedge[dash={1.5}{1.5}](D3,G3){}
\drawedge[dash={1.5}{1.5}](E3,G3){}
\drawedge[curvedepth=3](G3,C3){}
\drawedge[curvedepth=4](G3,B3){}
\node[Nadjust=w,Nadjustdist=1.5](H3)(36,34){\{1,2,3\}}
\drawedge[dash={1.5}{1.5},curvedepth=12](A3,H3){}
\drawedge[dash={1.5}{1.5}](B3,H3){}
\drawedge[dash={1.5}{1.5}](C3,H3){}
\drawedge[dash={1.5}{1.5}](F3,H3){}
\drawedge[curvedepth=1](H3,D3){}
\drawedge[curvedepth=4](H3,E3){}
\end{picture}
\caption{The graphs $\Gamma_k(\mathrsfs{E}_5)$ for $k=1,2,3$.}\label{fig:g123e5}
\end{center}
\end{figure}
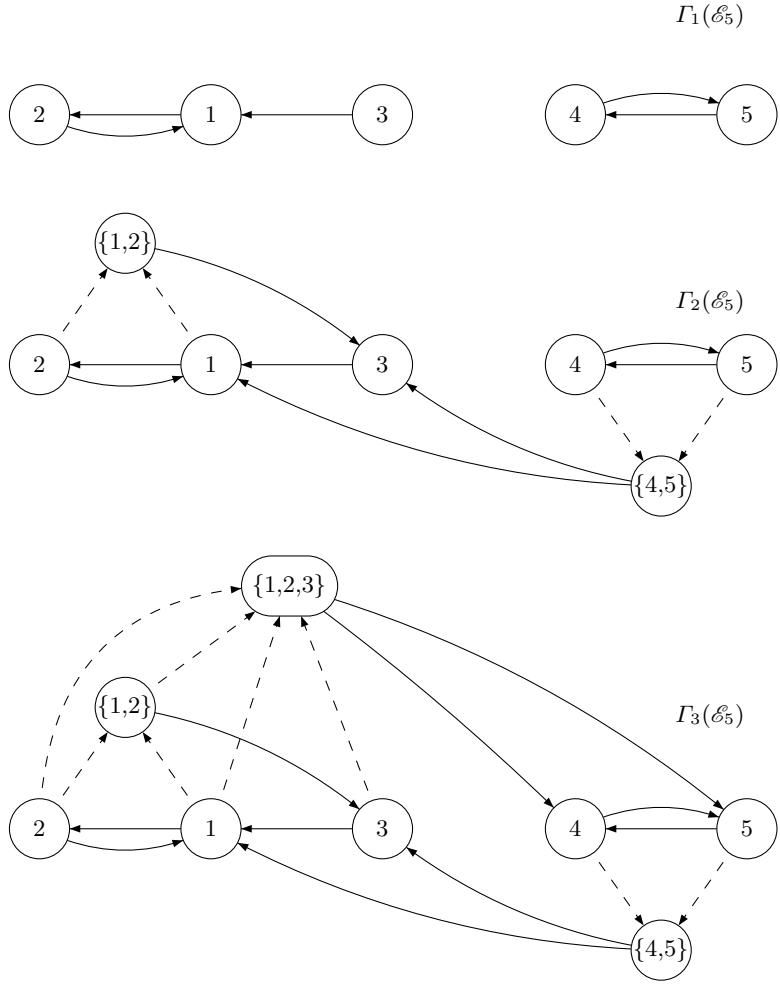
Fig.\,\ref{fig:g123e5} shows the graphs $\Gamma_1(\mathrsfs{E}_5)$ (on the top), $\Gamma_2(\mathrsfs{E}_5)$ (in the middle), and $\Gamma_3(\mathrsfs{E}_5)$ (on the bottom). The inclusion edges are shown with dashed arrows. One sees that $\Gamma_1(\mathrsfs{E}_5)$ has three \scc{}s whose supports are $\{1,2\}$, $\{3\}$, and $\{4,5\}$. The first and the last of these components have rank $\ge 2$. Hence their supports form the set $Q_2$ and occur as vertices of the graph $\Gamma_2(\mathrsfs{E}_5)$. This graph has two \scc{}s whose supports are $\{1,2,3\}$ and $\{4,5\}$. Only the first component has rank $\ge 3$ so it is its support that forms the set $Q_3$ and occurs as a vertex of the graph $\Gamma_3(\mathrsfs{E}_5)$. Since the latter graph is \scn, it yields $\Gamma(\mathrsfs{E}_5)$.
\end{example}

Back to the construction of the graph $\Gamma(\mathrsfs{A})$, observe that if $\mathrsfs{A}$ has $n$ states and the graph $\Gamma_k(\mathrsfs{A})$ is not \scn, then the maximal rank of its \scc{}s does not exceed $n-1$. Hence the number of steps in our process cannot exceed $n-1$. One can show that, for each $n$ and each $k\in\{1,\dots,n-1\}$, there exist automata with $n$ states (both completely reachable and not) such that the process terminates after exactly $k$ steps; we present such examples in Section~5.

\section{The main result}

We start with the promised generalization of Theorem~\ref{thm:sufficient2}.

\begin{theorem}
\label{thm:sufficient3}
If a DFA $\mathrsfs{A}=\langle Q,\Sigma\rangle$ is such that the graph $\Gamma(\mathrsfs{A})$ is \scn\ and $\Gamma(\mathrsfs{A})=\Gamma_k(\mathrsfs{A})$, then $\mathrsfs{A}$ is completely reachable; more precisely, for every non-empty subset $P\subseteq Q$, there is a product $w$ of words of defect at most $k$ such that $P=Q\dt w$.
\end{theorem}

\begin{proof}
As in the proof of Theorem~\ref{thm:sufficient2}, we take a non-empty subset $P\subseteq Q$ and proceed by induction on $m:=|Q{\setminus} P|$. The induction base $m=0$ is obvious, and  if $m>0$, then it suffices to find a subset $R\subseteq Q$ such that $P=R\dt w$ for some word $w$ of defect at most~$k$ and $|R|>|P|$. Then the induction assumption applies to $R$, and the same argument as in Theorem~\ref{thm:sufficient2} completes the proof.

Let $P$ be a proper subset of $Q$. If there exists an edge $(q,p)\in E_1$ such that $q\in Q{\setminus} P$ and $p\in P$, we can reuse the reasoning from Case~1 in the proof of Theorem~\ref{thm:sufficient2}. Thus, we may (and will) assume that for every edge $(q,p)\in E_1$, if $q\in Q{\setminus} P$, then also $p\in Q{\setminus} P$. Consider the set
\[
\widehat{P}:=P\cup\left\{C\in\bigcup_{i=2}^kQ_i\mathrel{\rule[-9pt]{.8pt}{24pt}} C\subseteq P\right\},
\]
which is a proper subset of $\widehat{Q}:=Q\cup\bigcup_{i=2}^kQ_i$ since $P$ is a proper subset of $Q$. Since the graph $\Gamma_k(\mathrsfs{A})$ is \scn, it has an edge $e$ that connects $\widehat{Q}{\setminus}\widehat{P}$ with $\widehat{P}$ in the sense that the head of the edge $e$ belongs to $\widehat{Q}{\setminus}\widehat{P}$ while the tail of $e$ lies in $\widehat{P}$. First, we show that $e$ cannot be an inclusion edge. Indeed, if $e$ is an inclusion edges, it is of the form either $(q,C)$, where $C\in\bigcup_{i=2}^kQ_i$ and $q\in C$ or $(D,C)$, where $D\in Q_i$, $C\in Q_j$ with $2\le i<j\le k$, and $D\subseteq C$. If the tail of $e$ lies in $\widehat{P}$, we have $C\subseteq P$ whence $q\in P$ for $e$ of the form $(q,C)$ and $D\subseteq P$ for $e$ of the form $(D,C)$. We see that the head of $e$ belongs to $\widehat{P}$ in either case, a contradiction.

Thus, every edge connecting $\widehat{Q}{\setminus}\widehat{P}$ with $\widehat{P}$ must belong to one of the sets $E_2,\dots,E_k$. Let $j$ be the least number such that $E_j$ contains an edge $(C,p_0)$ with $p_0\in P$ and $C\in\widehat{Q}{\setminus}\widehat{P}$. Then $C$ is the support of an \scc\ $\Delta$ of the graph $\Gamma_{j-1}(\mathrsfs{A})$.

We aim to show that $C\cap P=\varnothing$. Arguing by contradiction, assume that $C\cap P\ne\varnothing$. Since $C\notin\widehat{P}$, we have $C\nsubseteq P$. Thus, $C$ contains both states which are in $P$ and states which are not in $P$. Since $C$ is a part of $\Delta$ and $\Delta$ is \scn, any state in $C$ is reachable in the graph $\Gamma_{j-1}(\mathrsfs{A})$ from any other state in $C$. Take a directed path $\pi$ of minimum length in $\Gamma_{j-1}(\mathrsfs{A})$ such that its starting state lies in $C{\setminus}P$ while its terminal state lies in $C\cap P$. Consider the penultimate vertex of the path $\pi$. Would this vertex be some state $q\in Q$, we could conclude that $q\in P$: recall our assumption that for every $(q,p)\in E_1$, if $q\in Q{\setminus} P$, then also $p\in Q{\setminus} P$. Then removing the last edge from $\pi$ would yield a shorter path, still leading from a state in $C{\setminus}P$ to the state $q\in C\cap P$ (we have $q\in C$ since $q$ is a vertex in $\Delta$ and $C$ is the support of $\Delta$). This would contradict our choice of $\pi$. Thus, the penultimate vertex of $\pi$ is an element $D\in Q_i$ for some $i\le j-1$. All incoming edges for $D$ are inclusion edges, so it easily follows from the minimality of $\pi$ that the second last edge of the path $\pi$ must be of the form $(r,D)$ where $r\in D$. If $D\in\widehat{P}$, that is, $D\subseteq P$, then $r\in P$, and again, removing two last edges from $\pi$ would yield a shorter path, still leading from a state in $C{\setminus}P$ to the state $r\in C\cap P$ (as above, $r\in C$ since $C$ is the support of $\Delta$). We conclude that $D\in\widehat{Q}{\setminus}\widehat{P}$. Now, denoting the terminal state of the path $\pi$ by $p_1$, we see that the last edge of $\pi$, that is, the edge $(D,p_1)\in E_i$ connects $\widehat{Q}{\setminus}\widehat{P}$ with $\widehat{P}$. Since $i<j$, this contradicts our choice of the number $j$.

Thus, $C\cap P=\varnothing$. Now it is easy to complete the proof, following the final argument from Case~2 in the proof of Theorem~\ref{thm:sufficient2}. Indeed, by the definition of $E_j$, there exists a word $w$ of defect~$j$ with respect to $\mathrsfs{A}$ such that $C\supseteq\excl(w)$ and $p_0\in\dupl(w)$. By the definition of $\dupl(w)$, there exist some $q_1,q_2\in Q$ such that $q_1\ne q_2$ and $q_1\dt w=q_2\dt w=p_0$. Since $C\cap P=\varnothing$, we have $\excl(w)\cap P=\varnothing$. Therefore, for every state $r\in P{\setminus}\{p_0\}$, there is a state $r'\in Q$ such that $r'\dt w=r$. Now we set
$R:=\{q\mid q\dt w=p_0\}\cup\bigl\{r'\mid r\in P{\setminus}\{p_0\}\bigr\}$
and conclude that $P=R\dt w$ and $|R|>|P|$ since $q_1,q_2\in R$.\qed
\end{proof}

Theorem~\ref{thm:sufficient3} shows that a DFA $\mathrsfs{A}$ is completely reachable whenever the process of constructing the graph $\Gamma(\mathrsfs{A})$ terminates with SUCCESS. Our next result handles the case of FAILURE.

\begin{theorem}
\label{thm:necessary}
If a DFA $\mathrsfs{A}=\langle Q,\Sigma\rangle$ is such that the graph $\Gamma(\mathrsfs{A})$ is not \scn, then $\mathrsfs{A}$ is not completely reachable; more precisely, if $\Gamma(\mathrsfs{A})=\Gamma_k(\mathrsfs{A})$, then some subset in $Q$ with at least $|Q|-k$ states is not reachable in $\mathrsfs{A}$.
\end{theorem}

\begin{proof}
Assume that $\Gamma(\mathrsfs{A})=\Gamma_k(\mathrsfs{A})$ is not \scn. Then, by the construction of the graph $\Gamma(\mathrsfs{A})$, every \scc\ of $\Gamma(\mathrsfs{A})$ has rank $\le k$. The reachability relation on $\Gamma(\mathrsfs{A})$ induces a partial order on the set of its \scc{}s. Consider an \scc\ $\Delta$ which is minimal with respect to this partial order, and let $D$ be the support of $\Delta$. Let $\ell:=|D|$ and $P:=Q{\setminus}D$. As observed, we must have $\ell\le k$ whence $|P|=|Q|-\ell\ge |Q|-k$.

We aim to show that the subset $P$ is not reachable in  $\mathrsfs{A}$. Indeed, suppose that $P=Q\dt w$ for some word $w$ over the input alphabet of $\mathrsfs{A}$. The defect of $w$ is $\ell\le k$, whence the edges enforced by $w$ occur among the edges of $\Gamma(\mathrsfs{A})$. Let $(C,p)$ be an edge enforced by $w$. By the definition, the set $C$ must be the support of an \scc\ of the graph $\Gamma_\ell(\mathrsfs{A})$; besides that, the requirements  $C\supseteq\excl(w)=D$ and $p\in\dupl(w)\subseteq P$ must hold true. Now take any state $q\in D$. Since $D\subseteq C$, we have $q\in C$, and the inclusion edge $(q,C)\in I_\ell$ also occurs among the edges of $\Gamma(\mathrsfs{A})$ as $\ell\le k$. We see that the edges $(q,C)$ and $(C,p)$ form a directed path from $q$ to $p$ in $\Gamma(\mathrsfs{A})$. Recall that $\Delta$ is a minimal \scc\ whence every vertex of $\Gamma(\mathrsfs{A})$ reachable from a vertex in $\Delta$ must itself lie in $\Delta$. Thus, the state $p$ lies in the \scc\ $\Delta$ and hence belongs to its support $D$. On the other hand, we have $p\in P=Q{\setminus}D$, a contradiction.\qed
\end{proof}

Combining Theorems~\ref{thm:sufficient3} and~\ref{thm:necessary}, we readily  arrive at the following characterization of \cra, which is our main result:

\begin{theorem}
\label{thm:necess&suffic}
If a DFA $\mathrsfs{A}$ is completely reachable if and only if the graph $\Gamma(\mathrsfs{A})$ is \scn.
\end{theorem}

\section{Concluding remarks and future work}

\emph{The number of steps in the construction of $\Gamma(\mathrsfs{A})$.} For each $n\ge 2$, it is easy to find a DFA $\mathrsfs{A}$ with $n$ states and 2 input letters (both completely reachable and not) such that $\Gamma(\mathrsfs{A})=\Gamma_1(\mathrsfs{A})$. For \cra, one can use the \v{C}ern\'{y} automata $\mathrsfs{C}_n$, see, e.g., \cite[Example~1]{cra1}; for $\mathrsfs{A}$ being not completely reachable, one can take any permutation automaton with $n$ states.
Here we exhibit two series of DFAs with $n$ states, $\mathrsfs{E}_{n,k}$ and $\mathrsfs{E}'_{n,k}$, where $2\le k< n$, such that the DFAs $\mathrsfs{E}_{n,k}$ are completely reachable, the DFAs $\mathrsfs{E}'_{n,k}$ are not, and for each $\mathrsfs{B}\in\{\mathrsfs{E}_{n,k},\mathrsfs{E}'_{n,k}\}$, the construction of the graph $\Gamma(\mathrsfs{B})$ requires exactly $k$ steps.

The state set of both $\mathrsfs{E}_{n,k}$ and $\mathrsfs{E}'_{n,k}$ is the set $Q=\{1,2\dots,n\}$. The input letters of $\mathrsfs{E}_{n,k}$ are $a_{[1]},\dots,a_{[n]},a_{[1,n-k+1]},\dots,a_{[1,n-1]}$. We are going to define the action of the letters on $Q$; in this definition we denote $n-k+1$ by $\ell$. Now we set for each $i,j\in Q$,
\[
i\dt a_{[j]}:=\begin{cases}i&\text{if $i\ne j\le\ell$ or $i>j>\ell$},\\ i+1&\text{if $i=j<\ell$},\\ 1&\text{if $i=j=\ell$},\\
i-1&\text{if $i=j>\ell$},\\ i\dt a_{[j-1]}&\text{if $i<j$ and $j>\ell$}.\end{cases}
\]
Observe that the last line of the definition uses recursion. Besides that, for $i\in Q$ and $j=\ell,\dots,n-1$, we set $i\dt a_{[1,j]}:=\begin{cases}j+1&\text{if $i\le j$},\\ i&\text{if $i>j$}. \end{cases}$ For an illustration, consider the DFA $\mathrsfs{E}_3$ from Example~\ref{examp:not necessary}; it is easy to see that it belongs to the family $\mathrsfs{E}_{n,k}$, being nothing but its simplest member $\mathrsfs{E}_{3,2}$.

The DFA $\mathrsfs{E}'_{n,k}$ is obtained from $\mathrsfs{E}_{n,k}$ by omitting the letter $a_{[1,n-1]}$. Calculating the graph $\Gamma(\mathrsfs{B})$ for each $\mathrsfs{B}\in\{\mathrsfs{E}_{n,k},\mathrsfs{E}'_{n,k}\}$, one can see that $\Gamma(\mathrsfs{B})=\Gamma_k(\mathrsfs{A})$, and the process stops with SUCCESS for $\mathrsfs{E}_{n,k}$ and with FAILURE for $\mathrsfs{E}'_{n,k}$. By Theorem~\ref{thm:necess&suffic} the DFAs $\mathrsfs{E}_{n,k}$ are completely reachable, while the DFAs $\mathrsfs{E}'_{n,k}$ are not.

We omit the calculations due to page limit; in fact, they are straightforward because $\mathrsfs{E}_{n,k}$ and $\mathrsfs{E}'_{n,k}$ are designed such that every edge in the corresponding graphs is either an inclusion edge or enforced by a letter.

Observe that the size of the input alphabets of the DFAs $\mathrsfs{E}_{n,k}$ and $\mathrsfs{E}'_{n,k}$ is growing with $n$ and $k$. The question of whether or not similar series can be found among DFAs with restricted alphabets still remains open.

We mention in passing that the DFAs $\mathrsfs{E}_{n,k}$ with $k\le n-2$ answer (in the negative) a question asked at the end of \cite[Section~4]{cra1}: $\mathrsfs{E}_{n,k}$ with $k\le n-2$ is a \cran\ that induces no minimal \cran.

\smallskip

\noindent\emph{Converse statements for Theorems~\ref{thm:sufficient3} and~\ref{thm:necessary}.} We do not know if the converse of the ``more precise'' statement in Theorem~\ref{thm:sufficient3} holds true; that is, we do not know whether or not the condition that every non-empty subset is a reachable in $\mathrsfs{A}$ via a product of words of defect at most $k$ implies that the graph $\Gamma_k(\mathrsfs{A})$ is \scn. The question remains open even for $k=1$.

In contrast, it is easy to exhibit examples showing that the converse of the ``more precise'' statement in Theorem~\ref{thm:necessary} fails. In fact, for each $n>2$, there exists a DFA $\mathrsfs{A}$ with $n$ states in which some subset with $n-1$ states is not reachable but constructing the graph $\Gamma(\mathrsfs{A})$ requires $n-1$ steps.

\smallskip

\noindent\emph{Complexity issues.} Here we assume the reader's acquaintance with some basic concepts of computational complexity. Recall that it is still open whether or not complete reachability of a DFA can be recognized in polynomial time, see a detailed discussion in \cite[Section 3]{cra1}. It can be shown that the size of the graph $\Gamma(\mathrsfs{A})$ is polynomial in the size of $\mathrsfs{A}$ so that the result of the present paper directly implies that complete reachability can be decided in polynomial space. (This of course can be deduced from Savitch's Theorem as well.) Once the graph $\Gamma(\mathrsfs{A})$ is constructed, checking its strong connectivity in polynomial time makes no difficulty. However, it is far from being obvious that $\Gamma(\mathrsfs{A})$, even though it has polynomial size, can always be constructed in polynomial time.

Don \cite[Conjecture~18]{Don:2016} has formulated the following conjecture: if a completely reachable DFA $\mathrsfs{A}$ has $n$ states, then every subset of size $k$ can be reached in $\mathrsfs{A}$ via a word of length at most $n(n-k)$. Observe that if this conjecture holds true (even in the following weaker form: in a completely reachable DFA with $n$ states, every non-empty subset can be reached by a word of length $n^c$ for some constant $c$), then Theorem~\ref{thm:necess&suffic} implies that the problem of whether a given DFA $\mathrsfs{A}=\langle Q,\Sigma\rangle$ is completely reachable lies in the complexity class NP. Here is a brief outline of the proof of this claim; the details will be published elsewhere. By Theorem~\ref{thm:necess&suffic} one has to check that the graph $\Gamma(\mathrsfs{A})$ is \scn. Our non-deterministic algorithm guesses: the number $k$ that is expected to ensure the equality $\Gamma(\mathrsfs{A})=\Gamma_k(\mathrsfs{A})$; some subsets of $Q$ that are expected to serve as vertices in the set $Q_2\cup\cdots\cup Q_k$; and some words over $\Sigma$ of polynomial in $|Q|$ length that are expected to enforce enough edges of $\Gamma_k(\mathrsfs{A})$ to witness that the graph is \scn. Then one checks these data for consistency, and if they are consistent, builds in polynomial time the subgraph of $\Gamma_k(\mathrsfs{A})$ spanned by the edges enforced by the guessed words, together with inclusion edges between the guessed subsets. If the subgraph obtained this way is \scn, then so is $\Gamma_k(\mathrsfs{A})$.

Obviously, if $\Gamma(\mathrsfs{A})$ is indeed \scn, then the valid data described above exist, and the above algorithm has a chance to make correct guesses. The bottleneck here is the consistency check, as one has to calculate the sets $\excl(w)$ and $\dupl(w)$ for the guessed words $w$, and this can be done in polynomial in $|Q|$ time only under the assumption that the words have polynomial in $|Q|$ length.

\end{document}